\documentclass[conference]{IEEEtran}
\IEEEoverridecommandlockouts

\def\lefttag#1{\tag*{\makebox[0pt][l]{\hspace*{-\columnwidth}#1}}}

\usepackage{amsmath,amssymb,amsfonts,amsthm}
\usepackage{balance}

\usepackage{stfloats}

\usepackage[capitalize]{cleveref}
\crefname{algocf}{Algorithm}{Algorithms}

\usepackage[ruled, linesnumbered]{algorithm2e}
\SetKwInput{KwData}{Input}
\SetKwInput{KwResult}{Output}
\usepackage{tikz}
\usepackage{mathtools}

\newtheorem{theorem}{Theorem}
\newtheorem{lemma}[theorem]{Lemma}
\newtheorem{proposition}[theorem]{Proposition}
\newtheorem{corollary}[theorem]{Corollary}
\theoremstyle{definition}
\newtheorem*{remark}{Remark}

\usepackage{float}

\makeatletter
\def\?#1{}
\def\whp{w.h.p\@ifnextchar-{.}{\@ifnextchar.{.\?}{\@ifnextchar,{.}{\@ifnextchar){.}{\@ifnextchar:{.:\?}{.\ }}}}}}
\def\Whp{W.h.p\@ifnextchar.{.\?}{\@ifnextchar,{.}{.\ }}}
\makeatother

\renewcommand\log{\ln}

\newcommand{\vX}{\vec X}

\renewcommand{\epsilon}{\eps}

\newcommand\PSI{\vec\psi}

\renewcommand{\vec}[1]{\boldsymbol{#1}}

\newcommand\SIGMA{\vec\sigma}

\newcommand\G{\mathcal{G}}
\newcommand\cE{\mathcal{E}}

\newcommand\cS{\mathcal{S}}

\newcommand\cX{\mathcal{X}}

\def\cR{{\mathcal R}}
\def\cE{{\mathcal E}}

\newcommand\vy{\vec y}

\newcommand\eps{\varepsilon}

\newcommand\ZZ{\mathbb{Z}}
\newcommand\NN{\mathbb{N}}

\newcommand\Erw{\mathbb{E}}
\newcommand{\vecone}{\vec{1}}

\newcommand{\Bin}{{\rm Bin}}

\newcommand\bc[1]{\left({#1}\right)}
\newcommand\cbc[1]{\left\{{#1}\right\}}

\newcommand{\bck}[1]{\left\langle{#1}\right\rangle}
\newcommand\brk[1]{\left\lbrack{#1}\right\rbrack}
\newcommand\scal[2]{\bck{{#1},{#2}}}

\newcommand\abs[1]{\left|{#1}\right|}

\newcommand\RR{\mathbb{R}}

\newcommand\pr{\mathbb{P}} 
\renewcommand\Pr{\pr}

\newcommand{\floor}[1]{\left\lfloor#1\right\rfloor}

\newcommand\mqgtada{m^\text{\normalfont\scriptsize BPD}_{\text{\normalfont\scriptsize seq}}}
\newcommand\mqgt{m^\text{\normalfont\scriptsize BPD}_{\text{\normalfont\scriptsize para}}}

 \def\G{{\vec G}}

\def\ex{{\mathbb E}}
\def\pr{{\mathbb P}}

\newcommand{\remove}[1]{}

\newcommand\mgreedy{m_{\text{\textsc{mn}}}}

\let\oldparagraph\paragraph
\def\paragraph#1{\oldparagraph*{#1.\expandafter\?}}
\begin{document}
\makeatletter
\newcommand{\linebreakand}{%
  \end{@IEEEauthorhalign}
  \hfill\mbox{}\par
  \mbox{}\hfill\begin{@IEEEauthorhalign}
}
\makeatother
\title{On the Parallel Reconstruction from Pooled Data\thanks{OG and PL were supported by DFG CO 646/3. MHK was supported by DFG FOR 2975 and Stiftung Polytechnische Gesellschaft.}}
\author{\IEEEauthorblockN{Oliver Gebhard}
\IEEEauthorblockA{\textit{TU Dortmund University}\\
Dortmund, Germany \\
oliver.gebhard@tu-dortmund.de}
\\
\IEEEauthorblockN{Dominik Kaaser}%
\IEEEauthorblockA{\textit{Universität Hamburg}\\
Hamburg, Germany \\
dominik.kaaser@uni-hamburg.de}
\and
\IEEEauthorblockN{Max Hahn-Klimroth}%
\IEEEauthorblockA{\textit{TU Dortmund University}\\
Dortmund, Germany \\
max.hahn-klimroth@cs.tu-dortmund.de}
\\
\IEEEauthorblockN{Philipp Loick}%
\IEEEauthorblockA{\textit{Goethe University Frankfurt}\\
Frankfurt, Germany \\
loick@math.uni-frankfurt.de}

}

\maketitle
\thispagestyle{plain}
\pagestyle{plain}

\begin{abstract}
In the pooled data problem the goal is to efficiently reconstruct a binary signal from additive measurements.
Given a signal $\SIGMA \in \cbc{0,1}^n$, we can query multiple entries at once and get the total number of non-zero entries in the query as a result.
We assume that queries are time-consuming and therefore focus on the setting where all queries are executed in parallel. 
For the regime where the signal is sparse such that $\abs{\abs{\SIGMA}}_1 = o(n)$ our results are twofold:
First, we propose and analyze a simple and efficient greedy reconstruction algorithm.
Secondly, we derive a sharp information-theoretic threshold for the minimum number of queries required to reconstruct $\SIGMA$ with high probability.
Our first result matches the performance guarantees of much more involved constructions (Karimi et al.\ 2019).
Our second result extends a result of  Alaoui et al.\ (2014) and Scarlett \& Cevher (2017) who studied the pooled data problem for dense signals.
Finally, our theoretical findings are complemented with empirical simulations.
Our data not only confirm the information-theoretic thresholds but also hint at the practical applicability of our pooling scheme and the simple greedy reconstruction algorithm.
\end{abstract}


\section{Introduction}
We consider the \emph{binary pooled data problem with additive queries} which is defined as follows.
We are given a \emph{signal} of length $n$, a large vector $\SIGMA \in \cbc{0,1}^n$ of Hamming weight $k$ and a querying method.
Each query pools multiple entries of $\SIGMA$ together and returns the exact number of non-zero entries contained in the pool  (see \cref{fig:example} for an example).
The goal is to reconstruct $\SIGMA$ using as few queries as possible.

In many real-world scenarios the time to compute a reconstruction of $\SIGMA$ is dominated by the time to perform a single query.
The evaluation of such a query may require, e.g., computations using a deep neural network on a GPU \cite{liang2021neural}, biological processes such as DNA screening \cite{cao_2014, sham_2002}, or PCR tests in a bio-medical context \cite{BENAMI20201248}.
To obtain a substantial speed-up, we therefore focus on \emph{parallel} schemes where all queries are specified a priori and executed simultaneously.
{This assumption makes sense in the context of a life sciences laboratory: queries can be envisioned as measurements conducted by a liquid handling robot.
The time to perform all (parallel) queries then clearly dominates the time to run an efficient (sequential) reconstruction algorithm (for practical input sizes).
}

In this paper we focus on the \emph{sublinear regime} where the number of non-zero entries $k$ scales sub-linearly in the signal's length $n$ such that $k=n^{\theta}$ for some $\theta < 1$. 
In this setting, our main task is to specify a suitable parallel pooling design and an efficient reconstruction algorithm that allows us to compute $\SIGMA$ efficiently from the queried data.
We are interested in two different types of \emph{phase-transitions} that commonly arise in the analysis of reconstruction and statistical inference problems:
\begin{enumerate}
    \item What is the minimum number of queries that allows us to infer $\SIGMA$ from the query results given unlimited computational power?
    \item How many queries are required such that an efficient algorithm can compute $\SIGMA$ from the query results?
\end{enumerate}
\widowpenalty0
\clubpenalty0
We will refer to the first phase-transition as the \emph{information-theoretic threshold} and to the second phase-transition as the \emph{algorithmic threshold}.

\begin{figure}[t]
\centering
\begin{tikzpicture}[scale=1]
\node[circle, draw, minimum width=0.75cm, fill=lightgray] (x0) at (0, 0) {$\SIGMA_1$};
\node[circle, draw, minimum width=0.75cm, fill=lightgray] (x1) at (1,0) {$\SIGMA_2$};
\node[circle, draw, minimum width=0.75cm] (x2) at (2, 0) {$\SIGMA_3$};
\node[circle, draw, minimum width=0.75cm] (x3) at (3, 0) {$\SIGMA_4$};
\node[circle, draw, minimum width=0.75cm, fill=lightgray] (x4) at (4, 0) {$\SIGMA_5$}; 
\node[circle, draw, minimum width=0.75cm] (x5) at (5, 0) {$\SIGMA_6$};
\node[circle, draw, minimum width=0.75cm] (x6) at (6, 0) {$\SIGMA_7$};

\node[rectangle, draw, minimum width=0.75cm, minimum height=0.75cm] (a1) at (1, -2) {$2$};
\node[rectangle, draw, minimum width=0.75cm, minimum height=0.75cm] (a2) at (2, -2) {$2$};
\node[rectangle, draw, minimum width=0.75cm, minimum height=0.75cm] (a3) at (3, -2) {$3$};
\node[rectangle, draw, minimum width=0.75cm, minimum height=0.75cm] (a4) at (4, -2) {$1$};
\node[rectangle, draw, minimum width=0.75cm, minimum height=0.75cm] (a5) at (5, -2) {$1$};

\path[draw] (x0) -- (a1);
\path[draw] (x0) -- (a2);
\path[draw] (x0) -- (a3);
\path[draw] (x1) -- (a1);
\path[-] (x1) edge [bend left=5,dashed] (a3);
\path[-] (x1) edge [bend right=5,dashed] (a3);
\path[draw] (x2) -- (a1);
\path[draw] (x2) -- (a3);
\path[draw] (x2) -- (a2);
\path[draw] (x3) -- (a3);
\path[draw] (x3) -- (a4);
\path[draw] (x3) -- (a5);
\path[draw] (x4) -- (a2);
\path[draw] (x4) -- (a5);
\path[draw] (x4) -- (a4);
\path[draw] (x5) -- (a2);
\path[-] (x5) edge [bend left=5,dashed] (a4);
\path[-] (x5) edge [bend right=5,dashed] (a4);
\path[draw] (x6) -- (a3);
\path[draw] (x6) -- (a5);
\path[draw] (x6) -- (a4);
\end{tikzpicture}
\caption{A small example with signal $\SIGMA = (1,1,0,0,1,0,0) \in \cbc{0,1}^7$ at the top and queries $a_1,\dots,a_5$ at the bottom. The edges of the bipartite (multi-) graph $\G$ show which entries are contained in a specific query. The dashed lines highlight the occurrence of multi-edges. The goal is to reconstruct $\SIGMA$ given only $\G$ and the query results $(2, 2, 3, 1, 1)$.}
\label{fig:example}
\end{figure}

\subsection{The Teacher-Student Model}
As in many related reconstruction problems, the teacher-student model provides the fundamental means towards analyzing information-theoretic questions. 
The challenge in such reconstruction problems lies in deriving probability distributions that are dependent on a variety of random variables and hard to express per se. However, deriving probability distributions conditioned on certain high-probability events is feasible. 
For an introduction and mathematical justification of the model, we refer the reader to~\cite{aco_2018}.
The setup is the following: a teacher aims to convey some \emph{ground truth} to a student. Rather than directly providing the ground truth to the student, the teacher generates observable data from the ground truth via some statistical model and passes both the data and the model to the student. The student now aims to infer the ground truth from the observed data and the model.

In terms of this paper we see $\SIGMA$ as the ground truth. Its distribution is inherited from all vectors in $\cbc{0,1}^n$ of Hamming weight $k$. The observable data $\vec{y}$, together with the conducted queries (expressed as a graph $\G$) are passed to the student in order to infer $\SIGMA$.
In the following, we analyze the chances of the student to infer the ground truth from the observable data. First, we derive the model distribution from the provided information $\G$ and the query results $\vec{y}$. Afterwards, we use the gained knowledge to analyze the chances of the student to recover the ground truth by estimating the number of possible input vectors that are consistent with the observed query results. 
As our goal is to recover $\SIGMA$ with high probability, we condition on the event that the underlying bipartite multi-graph $\G$, which will be defined properly in due course, behaves almost \emph{as expected}. We exploit the knowledge about $\G$ to derive high-probability events which we can condition on. Eventually, our analysis conveys the information whether there is a unique input vector or multiple possible input vectors out of which the student has to guess the correct one.

\subsection{Related Work}
The binary pooled data problem, sometimes called \emph{quantitative group testing}, finds its roots in early works of Dorfman \cite{dorfman_1943}, Djackov \cite{djackov_1975}, and Shapiro \cite{shapiro_1960}.
It has recently gained a lot of interest in the literature
\cite{alaoui_2017, bshouty_2009,  feige2020quantitative, karimi_2019_2, scarlett_2017}, with applications in a multitude of disciplines such as DNA screening \cite{sham_2002}, identifying genetic carriers \cite{cao_2014} and machine learning \cite{liang2021neural, martins_2014, NIPS2014_fb8feff2}.
Variants of the problem include binary group testing \cite{AJS_book, coja_spiv} or threshold group testing \cite{Chan2013, DeMarco2020}.
We start our discussion with an overview of related work from information theory.

\paragraph{Information-Theoretic Aspects}
A simple information-theoretic lower bound can be obtained by a folklore counting argument: each query returns a number from $0$ to $k$, thus a pooling design with $m$ queries can produce at most $(k+1)^m$ different outcomes. This number must be larger than $\binom{n}{k}$ in order to distinguish all possible input vectors of length $n$ with Hamming weight $k$. By standard asymptotic bounds, we obtain
\begin{align}
    \mqgtada \geq \bc{1 - o(1)} \frac{\log \frac{n}{k}}{\log k} k. 
\end{align}
The universal lower bound on $\mqgtada$ holds in any case, even if the queries do not need to be conducted in parallel. Restricted to the important special case in which all queries are conducted in parallel, \cite{djackov_1975} shows that reconstruction of $\SIGMA$ requires at least 
\begin{align}\label{lowwer_seq}
    \mqgt = (2 - o(1)) \frac{\log \frac{n}{k}}{\log k} k = 2 \cdot \mqgtada
\end{align}
queries, even with unlimited computational power. On the positive side, Bshouty \cite{bshouty_2009} proves that reconstruction of $\SIGMA$ is efficiently possible with $(2 + \epsilon)\mqgtada$ queries if they are conducted sequentially and Grebinski and Kucherov \cite{grebinski_2000} provide a parallelizable design with an exponential-time reconstruction decoding algorithm which guarantees inference with $(2 + \epsilon) \mqgt$ queries using \emph{separating matrices}. {The latter positive result was extended to the so-called \emph{Subset Select problem} \cite{Marco_2013}, a relaxation of the pooled data problem that asks to identify only a subset of positive entries correctly. Recently, \cite{feige2020quantitative} improved the result for this relaxation by a factor of $2$.}
So far, these results hold independently of $k$. For the linear regime where $k = \Theta(n)$, much stricter results are already known: Alaoui et al.\ \cite{alaoui_2017} and Scarlett and Cevher \cite{scarlett_2017} show that there is an exponential-time construction that achieves reconstruction with $(1 + \epsilon)\mqgt$ parallel queries -- a result that is dependent on $k$ scaling linearly in $n$.

\paragraph{Algorithmic Aspects}
If allowed for sequential queries, Bshouty \cite{bshouty_2009} presents an efficient reconstruction algorithm that succeeds at recovery of $\SIGMA$ with no more than $(2 + o(1))\mqgtada$ queries.
However, for parallel schemes, there are significant gaps between the information-theoretic lower bound and the currently best known efficient algorithms \cite{alaoui_2017, donoho_l1, feige2020quantitative, Foucart2013, karimi_2019, OMP}. 
For instance, Alaoui et al.\ \cite{alaoui_2017} present an \emph{Approximate Message Passing} algorithm for dense signals ($k = \Theta(n)$).
Furthermore, Donoho and Tanner \cite{donoho_l1} give a decoding strategy based on $\ell_1$-minimization, and Foucart and Rauhut \cite{Foucart2013} introduce the \emph{Basis Pursuit}-algorithm. They can be used to recover $\SIGMA$ with
\[ (2  + o(1)) k \log \frac{n}{k}  \quad \text{and} \quad  (2 + o(1)) k \log n \sim \frac{2}{1 - \theta} k \log \frac{n}{k} \]
queries, respectively, if the signal is sparse ($k \ll n$). Note that these algorithms solve the more general compressed sensing problem. Various improvements over the Basis Pursuit algorithm are known (e.g., the Orthogonal Matching Pursuit \cite{OMP} and its improved version for discrete signals \cite{OMP_disc}) but as Wang and Yin \cite{Wang2010} discuss, they do not perform asymptotically better in the setting discussed in this paper.
More recent algorithms explicitly designed for recovery of $\SIGMA$ from additive queries in the sparse regime are due to Karimi et al.\ \cite{karimi_2019_2,karimi_2019}. They provide two algorithms based on graph codes that require  
 \[ (1.72 + o(1)) k \log \frac{n}{k} \quad \text{and} \quad (1.515 + o(1)) k \log \frac{n}{k} \]
queries, respectively.
Furthermore, in a yet unpublished draft that appeared subsequently to our work on arXiv, Feige and Lellouche \cite{feige2020quantitative} analyze the Subset Select problem. They prove that, under mild assumptions, an algorithm succeeding at this relaxation can be turned into an algorithm for recovery of $\SIGMA$ without significantly increasing the required number of queries.

\subsection{Our Contributions}
We study the pooled data problem under the random regular model $\G$ which is known to be information-theoretically optimal in the linear regime as well as in similar inference problems \cite{coja_spiv}.
More precisely, we let $\G = (V \cup F, E)$ be a random bipartite multi-graph with \emph{query-nodes} $F = \{a_1, \ldots, a_m\}$ representing the queries, \emph{entry-nodes} $V = \{x_1, \ldots, x_n\}$ representing the coordinates of $\SIGMA$, and edges $E$ indicating how often a specific entry is contained in a given query. Hereby, each query $a_i \in F$ contains exactly 
$\Gamma = n/2$ entries chosen uniformly at random with replacement. 

\paragraph{Algorithmic Results}
For the aforementioned pooling design we present a fairly intuitive greedy algorithm called \emph{Maximum Neighborhood (MN) Algorithm} that allows reconstruction of $\SIGMA$ \whp.\footnote{%
The expression \emph{with high probability (w.h.p.)} refers to a probability that tends to 1 as $n \rightarrow \infty$.%
} It follows a thresholding approach that is much simpler than the known algorithms by Karimi et al.\ \cite{karimi_2019_2,karimi_2019}, which are technically highly challenging. A formal definition of the MN-Algorithm is given in \cref{MN_algorithm}.

\begin{algorithm}[ht]

\DontPrintSemicolon
\SetAlgoVlined
\SetKwFor{ForParallel}{for}{do in parallel}{end}
\SetKw{KwTo}{to}
\SetKwFunction{Query}{query}

\KwData{$m$, $k$, querying method \Query}
\KwResult{estimation $\Tilde{\SIGMA}$ for $\SIGMA$.}
\medskip
\ForParallel{$i = 1$ \KwTo $m$}{
    sample a multiset $a_i$ of size $\Gamma$ from $[n]$\;
    compute $\vec{y}_i \gets $ \Query{$a_i$}\;
    \tcp{\normalfont\small The query method guarantees that $\vec{y}_i = \sum_{j \in a_i} \SIGMA(j)$.}
}
\medskip
\For{$i = 1$ \KwTo $n$}{
calculate $\Psi_i \gets \sum_{j = 1}^m \vecone \cbc{i \in a_j} \cdot \vec{y}_j$\;
calculate $\vec \Delta^\star_i \gets \sum_{j = 1}^m \vecone \cbc{i \in a_j}$\;
}

sort coordinates of $\Tilde{\SIGMA}$ in decreasing order by $\Psi_i - \vec \Delta^\star_i \frac{k}{2}$\;
set $\Tilde{\SIGMA}$ to $1$ for the first $k$ (sorted) coordinates\;
set $\Tilde{\SIGMA}$ to $0$ for the remaining $n-k$ (sorted) coordinates\;
\caption{The Maximum Neighborhood Algorithm}
\label{MN_algorithm}
\end{algorithm}

On an intuitive level, the MN-Algorithm works as follows.
First, we query $m$ times exactly $\Gamma$ randomly chosen entries of the signal in parallel, which yields the graph representation $\G$.
Secondly, we sum up the query results $(\vec y_a)_{a \in \partial x}$ in the neighborhood induced by $\G$ of each coordinate, counting multi-edges only once.
The sum is then centralized by its expected value.
Finally, those coordinates with a large \emph{score} are very likely to have the value $1$ under $\SIGMA$.
Our first main theorem states how many parallel queries are required for the MN-Algorithm to recover the correct $\SIGMA$ \whp.

\begin{theorem}
\label{Thm_alg}
Suppose that $0<\theta < 1$, $k = n^\theta$, and $\epsilon>0$ and let
\begin{align*}
     \mgreedy(n, \theta) = 4 \left(1-\frac{1}{\sqrt{e}}\right) \frac{1 + \sqrt{\theta}}{1 - \sqrt{\theta}} k \log(n/k).
\end{align*}
If $m>(1+\epsilon) \mgreedy(n,\theta)$, then \cref{MN_algorithm} outputs $\SIGMA$ \whp on input $m$ and $k$ and an additive querying method \texttt{query} that returns the total number of one-entries in a query.
\end{theorem}
While the MN-algorithm takes $k$ as an input, the proof reveals that prior knowledge of $k$ is not required in detail. More precisely, a lower bound on $k$ suffices, as in this case enough queries are conducted and the design of $\G$ is independent from $k$. Observe that one additional parallel query on all entries reveals the exact value of $k$ immediately without increasing $m$ asymptotically and therefore the only dependence on $k$ in \cref{MN_algorithm} (Line 7) can be easily removed by this one additional query. Beside not being strictly dependent on $k$, a main novelty of the MN-algorithm is its greedy fashion, providing a straightforward approach compared to the technically challenging algorithms presented in \cite{karimi_2019_2,karimi_2019}.

\paragraph{Parallelized Reconstruction}
Observe that our reconstruction algorithm, apart from sampling the test design and performing all queries in parallel, is specified in a sequential fashion.
This emphasizes the local structure of the reconstruction algorithm.
In the context of a parallel computation we observe that our algorithm can be readily parallelized.
When individual queries can be conducted much faster, this further reduces the overall running time of our approach.
Such improved reconstruction algorithms can be used in the context of machine learning, see, e.g., \cite{NIPS2014_fb8feff2} for an application.

Recall that our test design is described by a random bipartite graph $\G$ and let $M = M(G) = ( m_{ij} ) \in \{0,1\}^{n\times m}$ be the unweighted biadjacency matrix of $G$.
Intuitively, the entries of $M$ are those values that are summed up in Line $6$ of \cref{MN_algorithm}.
It follows that the $\Psi_i$ and $\vec \Delta^\star_i$ vectors are matrix-vector products $\vec \Delta^* = M 1$ and $\Psi = M \vec y$ where $1 = (1, \dots, 1)$ is the all-one-vector and $\vec y$ is the query result vector.
The sums computed in Lines 4 to 6 of \cref{MN_algorithm} can therefore be expressed in terms of two matrix-vector products for which efficient parallelizations are known.
Finally, in Lines 7 to 9 of \cref{MN_algorithm} the (coordinates of) the resulting vector are sorted.
See \cite{DBLP:journals/ijpp/SinghJC18} for a rather recent survey (with a focus on but not limited to GPUs) on parallel sorting algorithms.
\color{black}

\paragraph{Information-Theoretic Results}
We prove that in the sublinear regime where $k = n^\theta$ for some $\theta \in (0,1)$ it is possible to reconstruct $\SIGMA$ from $(\G, \vec y)$ with high probability with no more than $(1 + \epsilon) \mqgt$ parallel queries for some arbitrarily small $\eps > 0$. More precisely, we show that there is, with high probability, no second input vector $\tau \in \{0,1\}^n$ leading to the same sequence of query results. 
\begin{theorem}\label{Thm_inf}
Suppose that $0<\theta< 1$, $k = n^\theta$, and $\epsilon>0$ and let
\begin{align*}
 \mqgt &= 2 \frac{k \log(n/k)}{\log k} = 2 \frac{1 - \theta}{\theta} k.
\end{align*}
If $m> (1+\epsilon)\mqgt$, $\SIGMA$ can be computed from $\G$ and $\vec{y}$ \whp.
\end{theorem}
Our result reduces the previously known upper bound of Grebinski and Kucherov \cite{grebinski_2000} by a factor of two and we provide the missing counter part of \eqref{lowwer_seq} which establishes the existence of a phase-transition at $\mqgt$ for parallel designs.

\subsection{Discussion}
Our results extend information-theoretic results of Alaoui et al.\ \cite{alaoui_2017} from the linear regime to the sublinear regime. For $\theta \to 1$, our threshold of \cref{Thm_alg} turns out to converge towards the threshold of \cite{alaoui_2017}.
The study of the sublinear regime is inspired by studies of the compressed sensing problem with a sparse underlying signal \cite{arjoune_2017}.
In the special case of the binary pooled data problem, those studies were initiated by \cite{karimi_2019}.
The sparse regime is indeed interesting in real-world applications, with examples including epidemiology where Heaps law models the early spread of pandemics \cite{benz_2008,wang_2011} or the detection of rare features in image classification in machine learning \cite{liang2021neural}.
{The relevance of the sublinear regime can be seen in the following example. Suppose a screening for HIV is conducted. Out of about 67,220,000 residents of the UK, 105,200 are known to be infected with the HI virus. Hence, by screening n = 10.000 random probes, we expect 16 positive entries in the signal corresponding to the infection status. Thus, the choice $\theta = 0.3$ describes the situation quite well.}

It is not surprising that also similar problems have been recently analyzed in the sublinear regime.
By now, a vast body of related literature exists (see, e.g., the survey by Aldridge et al.\ \cite{AJS_book}).
Interestingly, for the (presumably more difficult) variant in which a query only returns the information whether at least one non-zero entry was found, a very sophisticated efficient algorithm is known for $\theta \leq \log 2 / (1 + \log 2) \approx 0.409$ which requires $m_{GT} \sim \log^{-1}(2) k \log \frac{n}{k}$ parallel queries \cite{coja_spiv}. Thus, dropping most of the available information and using this approach outperforms not only the simple greedy approach discussed in this paper for small values of $\theta$, but also the quite involved algorithms by Karimi et al.\ \cite{karimi_2019_2, karimi_2019}. This result is of fundamental theoretical interest, since it solves an open complexity theoretical question. Nevertheless, their proposed algorithm appears to be of rather limited interest for practical applications, as it requires, e.g., that $\sqrt{\log \log n}$ is large.
This is in contrast to our simple greedy scheme, which our simulations have shown to work well for real-world input sizes.

As in state-of-the art designs for similar reconstruction problems \cite{AJS_book, coja_spiv}, we allow a specific entry to be included multiple times in one query. While this seems counter-intuitive in the first place, it does not affect practicability of the proposed design.

\section{Model and Notation} \label{sec:outline}
In this section we formally introduce the pooling design. As before, $\SIGMA \in \cbc{0,1}^n$ is the ground truth chosen uniformly at random from all $0-1$ vectors of length $n$ with exactly $k$ non-zero entries, where $k = n^{\theta}$ for some $\theta \in (0,1)$.
We use $\G=\G(n,m,\Delta)$ to denote the random bipartite multi-graph that models the pooling design, where $m$ denotes the total number of queries and  $\vec \Delta=\{\vec \Delta_1, \dots, \vec \Delta_n\}$ describes the number of queries each individual participates in.
Observe that $\vec{\Delta}_i \sim \Bin(mn/2,1/n)$.
Similarly, we let $\vec \Delta^{\star}=\{\vec \Delta^{\star}_1, \dots, \vec \Delta^{\star}_n\}$ denote the number of \emph{distinct} queries with expected value $\Erw \brk{ \vec \Delta_i^\star } = (1 - \exp(-1/2))m$.
We let the vector $\vec{y}\in\cbc{0, \dots, \Gamma}^m$ denote the sequence of query results. When we refer to any other input vector than $\SIGMA$, we simply write $\sigma$ for the input vector and $y=y(\G, \sigma)$ for the corresponding results' vector.
Additionally, we write $V=\cbc{x_1,\ldots,x_n}$ for the set of the $n$ entries of $\SIGMA$ and let
$V_0=\cbc{x_i\in V:\SIGMA(i)=0}$ and $V_1=V\setminus V_0$ be the set of entries with value $0$ and $1$, respectively.
For $x_i\in V$, we write $\partial x_i$ for the multiset of queries $a_j$ in which $x_i$ is contained.  Similarly, we write $\partial^{\star} x_i$ for the set of \emph{distinct} such queries.
Analogously, for a query $a_i$, we denote by $\partial a_i$ the multiset of entries that are contained.

Recall that in our model every query contains exactly $\Gamma=n/2$ entries, and those entries are assigned uniformly at random with replacement. If a one-entry $x_i$ participates in a query $a_j$ more than once, it increases $\vec{y}_j$ multiple times.
For each $x_i \in V$, we let $\Psi_i$ be the sum of its query results for \emph{distinct} queries it belongs to. That is, even if the entry appears more than once in a query and thus contributes to the result multiple times, this query's result contributes to $\Psi_i$ only once.
Of course, the value of $x_i$ under $\SIGMA$ has a significant impact on this sum, increasing it by $\vec \Delta_i$, if $x_i$ is non-zero. To account for this effect in our analysis, we introduce a second variable $\Phi_i$ that sums all the query results in which $x_i$ is contained and excludes the impact of $x_i$. Formally, for any configuration $\sigma \in \{0,1\}^n$ we define
\begin{align*} 
\Psi_i(\sigma) = \sum_{j \in \partial^{\star} x_i} y_{a_j} \quad \text{and} \quad \Phi_i(\sigma) = \Psi_i(\sigma) - \vecone{\{\sigma(i)=1\} }\Delta_i
\end{align*}
and let $\Psi=(\Psi_1,\dots,\Psi_n)$ and $\Phi=(\Phi_1,\dots,\Phi_n)$. When we consider a specific instance $(\G, \vec{y})$, we will write $\vec{\Psi}_i = \Psi_i(\SIGMA)$ and $\vec{\Phi}_i = \Phi_i(\SIGMA)$ for the sake of brevity. Notably, while $\vec{\Psi}_i$ is known to the observer or an algorithm instantly from the queries, $\vec{\Phi}_i$ is not, since the ground truth $\SIGMA$ itself is unknown.

To express the number $m$ of queries conducted, we let $c(n) > 0$ denote a positive function from $\mathbb{N}$ to $\RR^+$ such that \begin{align*}
m &= c(n) k \frac{\log(n/k)}{\log k} .
\end{align*}
While it turns out that $c(n)=\Theta(1)$ suffices in the analysis of the information-theoretic bound, we will see that the performance guarantee of the MN-algorithm requires $c(n)$ to scale as $\Theta(\log n)$. 
Finally, we define a high probability event $\cR$ that we will condition on as explained in the teacher-student model. 
Let $\cR$ be the event that, for all $i \in [n]$, we have
\begin{equation}\label{def_R} 
\begin{aligned}
    \vec{\Delta}_i &= \frac m2 + O\bc{\sqrt{m} \log n} \\
    \text{and}\quad \vec{\Delta}_i^{\star} &= \bc{1-\exp \bc{-1/2}}m + O\bc{\sqrt{m} \log n}, 
\end{aligned}
\end{equation}
meaning that the underlying random graph satisfies concentration properties.
The following lemma states that $\cR$ is indeed a high probability event.
\begin{lemma}
\label{lem_event_R_high_prob}
If $\G$ is constructed according to our pooling scheme, then $\Pr(\cR) = 1 - o(1)$.
\end{lemma}
The proof follows from standard concentration results, see the appendix for the technical details.
Since \cref{Thm_inf,Thm_alg} only contain \whp-assertions, we can safely condition on $\cR$ for the remainder of our analysis.

\section{MN-Algorithm}
\paragraph{Outline}\label{perf_guarantees}
Recall that $\Psi_i$ is the sum over all query results in which the entry $x_i$ is contained (multi-edges counted only once) and $\vec \Delta^\star_i$ is the (random) number of disjoint such queries. Furthermore, let $\cE_j$ be the $\sigma-$algebra generated by the edges connected with $x_j$. As already discussed, we get 
\begin{equation*}\vec \Delta^\star_i = (1 + o(1)) \bc{1 - \exp \bc{-1/2}}m\end{equation*} 
\whp. Therefore, intuitively spoken, a non-zero entry $x_i$ increases the value of $\vec \Psi_i$ by $\vec \Delta_i = (1 + o(1))m/2$, other than zero-entries. Moreover, by construction of the random bipartite (multi-)graph $\G$, we get that the second neighborhood of $x_i$ contains $\Bin \bc{ \Gamma \vec \Delta^\star_i, k/n }$ non-zero entries. Thus we expect 
\[ \Erw \brk{ \Psi_i - \vec \Delta^\star_i \frac{k}{2} \Big| \cE_i} = \vecone {\cbc{ \SIGMA(i) = 1 }} \vec \Delta_i. \]
Therefore, if $\Psi_i - \vec \Delta^\star_i \frac{k}{2}$ is called the \emph{score} of entry $x_i$, we observe that the scores differ between zero entries and non-zero entries.
The whole proof of the algorithmic performance boils down to identify a threshold value $T(\alpha) = T(n, k, \alpha)$ such that, if sufficiently many queries are conducted, all scores of zero entries  are below $T(\alpha)$ while the scores of all non-zero entries exceed this threshold \whp. If we conduct $m = d k \log \frac{n}{k}$ queries, with $d=c(n)\log(k)^{-1}$, we get by a standard application of a Chernoff bound and a union bound over all $k=n^\theta$ non-zero entries $x_i \in V_1$ and, respectively, $n-k=\Theta(n)$ zero-entries $x_i \in V_0$ that $T(\alpha)$ is a valid threshold whenever
\begin{equation}
\begin{aligned}
    \label{eqs_greedy}
    &\frac{- (1-\theta)\alpha^2 d}{4  \bc{1 - \exp \bc{-1/2}}(1 + o(1))}+\theta <0 \\\text{and}\quad  &\frac{- (1-\theta)(1-\alpha)^2 d}{4  \bc{1 - \exp \bc{-1/2}} (1 + o(1))}+1< 0,
\end{aligned}
\end{equation}
which will become clear in a second. Optimizing \eqref{eqs_greedy} with respect to $\alpha \in (0,1)$ and plugging $d$ into $m=dk\log(n/k)$ yields for any $\epsilon > 0$ the sufficient condition
\begin{equation*} m \geq (4 + \epsilon)(1 + o(1)) \bc{1-\exp \bc{-1/2}} \frac{1 + \sqrt{\theta}}{1 - \sqrt{\theta}} k \log(n/k). \end{equation*}

\paragraph{Formal Analysis}\label{sec:proof_algo}
Let $\vec A_{ij} \in \NN_0$ denote how often entry $x_i$ appears in query $a_j$ and let $\vec A = \bc{ \vec A_{ij} }_{i \in [n], j \in [m]}$ be the adjacency matrix of $\G$. Then the following holds.
\begin{corollary}\label{Fact_bin}
Let $1\leq j\leq n$.
Given $\cE_j$, the random variable
\begin{align*}
\vec S_j&=\PSI_j-\vec\Delta_j=\sum_{i=1}^m\vecone\cbc{ \vec A_{ij} > 0 }\bc{\vy_j-\vec A_{ij}}
\end{align*}
has distribution $\displaystyle\Bin\bc{\vec \Delta^\star_j \Gamma-\vec\Delta_j,\frac{k-\vecone\{ \SIGMA(j) = 1\} }{n-1}}.$
\end{corollary}
\begin{proof}
This is an immediate consequence of the model definition. There are $\Gamma \vec \Delta^\star_j - \vec \Delta_j$ half-edges connected to query-nodes in the neighborhood of $x_j$ that are connected to entry-nodes $x_i \neq x_j$. Each of these half-edges is connected to one of $k - \vecone \cbc{\SIGMA(j) = 1}$ entry-nodes belonging to an entry of value $1$, independently, from the $n-1$ remaining entry-nodes. 
\end{proof}
Now it is possible to immediately infer the expectation of $\vec S_j$ conditioned on the event $\cR$ (as defined in \eqref{def_R}). For the sake of brevity let $\gamma = 1 - \exp \bc{-1/2}$.
Given the event $\cR$ which guarantees concentration properties of the underlying graph, we get \whp
\begin{align}
    \label{eq_delta} \ex\brk{ \vec S_j\mid\cE_j, \cR } &= \bc{1 \pm \delta} \frac{\gamma k m}{2} \\ \lefttag{where} \notag \delta &:= \frac{\sqrt{2} \log n}{\sqrt{\gamma m k}} = o(1).
\end{align} 
The Chernoff bound allows us to bound $\vec S_j$ as follows. 
\begin{lemma}\label{Lemma_tails}
Let $\alpha \in (0,1)$ be a constant and $m = d k \log \frac{n}{k}$. Then \begin{align*}
	\MoveEqLeft \pr\bc{\abs{\vec S_j-\ex\brk{\vec S_j\mid\cE_j, \cR}} \geq (1 - \alpha) m/2\mid\cE_j, \cR}\\
	&\leq \exp \bc{ - \bc{1+o(1)} \frac{(1 - \alpha)^2 d}{4 \gamma (1 + o(1))} \log \frac{n}{k}}.
\end{align*}
\end{lemma}

\begin{proof}
The Chernoff bound (\cref{Lem_ChernoffBoundDKL}) directly implies 
\begin{align*}
	\MoveEqLeft \pr \bc{\abs{\vec S_j-\ex\brk{\vec S_j\mid\cE_j}\mid\cE_j, \cR}\geq(1 - \alpha) m/2 \mid \cR} \\
	&\leq\exp\bc{- \bc{1+o(1)} \frac{(1 - \alpha)^2m}{8 \Erw \brk{\vec S_j \mid \cE_j, \cR}} } \\
	&=\exp\bc{- \bc{1+o(1)}\frac{(1 - \alpha)^2 d }{4\gamma(1 + o(1))}\log(n/k)}. \qedhere
\end{align*}
\end{proof}
Next we show that, with a suitable choice of a threshold, the scores of zero- and one-entries are well separated. \begin{corollary}\label{Cor_tails} Let $\epsilon > 0$ be an arbitrary constant. If $m \geq \bc{4 + \epsilon} \bc{1 - \exp \bc{-1/2}} \frac{1 + \sqrt{\theta}}{1 - \sqrt{\theta}} k \log \frac{n}{k}$ then there exists an $\alpha \in (0,1)$ such that, \whp, we have
\begin{align*}
	\vec S_j+\vec\Delta_j&\geq\ex\brk{\vec S_j\mid\cE_j, \cR}+(1-\alpha)m/2
\intertext{for all $x_j$ where $\SIGMA(j) = 1$, and}
	\vec S_j&<\ex\brk{\vec S_j\mid\cE_j, \cR}+(1-\alpha)m/2
\end{align*}
for all $x_j$ where $\SIGMA(j) =0$.
\end{corollary}
\begin{proof}
	Let $x_j \in V_1(\G)$. Again, we make use of the concentration properties guaranteed by conditioning on $\cR$. Therefore, we assume that $\vec\Delta_j=m/2+O(\sqrt m\log n)$.
	Then \cref{Lemma_tails} ensures that
	\begin{align*}
		\MoveEqLeft \pr\bc{\vec S_j+\vec\Delta_j\leq\ex\brk{\vec  S_j\mid\cE_j, \cR}+(1-\alpha)m/2\mid\cE_j, \cR}\\
		&\leq
		\exp \bc{-\alpha^2 d /(4 \gamma (1 + o(1))) \log \frac{n}{k}}\\
		&=\exp\bc{ \frac{(\theta-1)\alpha^2 d}{4 \gamma (1 + o(1))}\log n}.
	\end{align*}
	Hence, the union bound shows that the first inequality holds for all $k$ elements of $V_1(\G)$ \whp if
		\begin{align}\label{eqtail1}
		\frac{(\theta-1)\alpha^2d}{4 \gamma (1 + o(1))}+\theta&<0.
	\end{align}
	Analogously, the second inequality holds for all $n-k$ elements of $V_0(\G)$ \whp if
	\begin{align*}
		\MoveEqLeft \pr\brk{\vec S_j\geq\ex\brk{\vec S_j\mid\cE_j, \cR}+(1-\alpha)m/2\mid\cE_j, \cR}\\
		&\leq
		\exp \bc { \bc{(1-\alpha)^2d/(4 \gamma (1 + o(1)))} \log \frac{n}{k}}\\  &=\exp\bc{\frac{(\theta-1)(1-\alpha)^2d}{4 \gamma (1 + o(1))}\log n}.
	\end{align*}
	Again, the union bound shows that the second inequality holds \whp if
\begin{align}\label{eqtail2}
	\frac{(\theta-1)(1-\alpha)^2d}{4 \gamma (1 + o(1))}+1<0.
	\end{align}
	
	Note that the condition in \eqref{eqtail1} is monotonically decreasing in $\alpha$ while the condition in \eqref{eqtail2} is monotonically increasing in $\alpha$.
	Hence the optimal choice of $\alpha$ is the one that makes the two terms in \eqref{eqtail1} and \eqref{eqtail2} equal:
	\begin{align*}
	\frac{(\theta-1)\alpha^2d}{4 \gamma (1 + o(1))}+\theta=\frac{(\theta-1)(1-\alpha)^2d}{4 \gamma (1 + o(1))}+1,
	\end{align*}
	which boils down to
	\begin{align*}
	\alpha&=\frac{d-4 \gamma (1 + o(1))}{2d}.
	\end{align*}
 By putting this solution for $\alpha$ into \eqref{eqtail1} we get
	\begin{align*}
	\frac{(\theta-1)(d-4 (\gamma + o(1))^2}{16\gamma d + o(1)}+\theta < 0.
	\end{align*}
	It now suffices to find the minimal $d=d(\theta) > 0$ such that
	\begin{align*}
	\frac{(\theta-1)(d-4\gamma + o(1))^2}{16\gamma d + o(1)}+\theta=0.
	\end{align*}
	Hence, we solve for (positive) $d$ and obtain that \cref{eqtail1,eqtail2} hold \whp provided 
	\begin{align*}
	d\geq4\gamma\cdot\frac{1+\sqrt{\theta}}{1-\sqrt{\theta}} + o(1),
	\end{align*}
	which matches the assumption in the lemma statement.
\end{proof}
We are now ready to formally prove \cref{Thm_alg}.
\begin{proof}[Proof of \cref{Thm_alg}]
According to \cref{lem_event_R_high_prob}, the event $\cR$ is a high-probability event.
\cref{Cor_tails} then immediately implies the theorem, together with the definition $m = d k \log \frac{n}{k}$.
\end{proof}

\section{Information-Theoretic Achievability}
In the following section we prove \cref{Thm_inf}.
Our approach is based on counting alternative input vectors $\sigma \neq \SIGMA$ that yield the same sequence of query results as the ground truth $\SIGMA$.
Note that the underlying techniques are regularly employed for random constraint satisfaction problems \cite{aco_2018}.

We start with an outline of the proof.
Let $S_k(\G,\vec{y})$ be the set of all vectors $\sigma\in\{0,1\}^n$ of Hamming weight $k$ such that 
\begin{align*}
\vec{y}_{a_i}&=\abs{\cbc{x_j \in \partial a_i: \sigma(j) = 1}}\quad\text{for all $i\in[m]$}.
\end{align*}
This means, we fix $m$ queries $a_1, \ldots, a_m$ and let $S_k(\G,\vec{y})$ be the set of all vectors $\sigma \in \cbc{0,1}^n$ with exactly $k$ ones that are consistent with the query results.
Let now $Z_k(\G,\vec{y})=|S_k(\G,\vec{y})|$.
We need to prove that $Z_k(\G,\vec{y}) = 1$ \whp if the number of queries $m$ exceeds $\mqgt$.
Note that we can always reconstruct $\SIGMA$ exactly in this case via an exhaustive search (recall that from an information-theoretic point of view the computational power is assumed to be unlimited).

In our analysis, it turns out that it is much more convenient to study $Z_{k,\ell}(\G, \vec{y})$, the number of alternative vectors that are consistent with the query results and have a so-called \emph{overlap} of $\ell$ with $\SIGMA$.
The overlap is the number of one-entries under $\SIGMA$ that are also present in an alternative vector $\sigma$. 
Formally, we define
\begin{align*}
Z_{k,\ell}(\G,\vec{y})&=\abs{\cbc{\sigma\in S_k(\G,\vec{y}):
\sigma \neq \SIGMA, \scal{\SIGMA}\sigma=\ell}}.
\end{align*}
It now suffices to prove that $\sum_{\ell = 0}^{k-1} Z_{k,\ell}(\G,\vec{y})=0$ for $m \geq (1+\eps) \mqgt$ \whp.
To this end, two separate arguments are needed.
First, we show in \cref{Prop_small_overlap} via a first moment argument that no second satisfying input vector $\sigma$ can exist with a small overlap with $\SIGMA$.
Secondly, we employ in \cref{Prop_big_overlap} the classical coupon collector argument to show that a second satisfying configuration cannot exist for large overlaps. 
Intuitively, this means that an entry that is flipped from zero under $\SIGMA$ to one under an alternative configuration $\sigma$ initiates a cascade of other changes to maintain the observed query results. 
The full technical proofs for the following statements can be found in the appendix.

\begin{proposition}\label{Prop_small_overlap}
Let $\epsilon>0$, $0<\theta< 1$ and assume that $m>(1+\epsilon)\mqgt$. \Whp, we have \begin{equation*}
\sum_{\ell = 0}^{k(1 - \exp(-1/2))}Z_{k,\ell}(\G,\vec{y})=0.
\end{equation*}
\end{proposition}
We now sketch the proof of \cref{Prop_small_overlap}.
By Markov's inequality it suffices to show that $\Erw[Z_{k,\ell}(\G,\vec{y})] \to 0$ fast enough for all $\ell$ with $0 \leq \ell < k - \bc{1 - \exp(-1/2)} \log k$ if $m \geq (1 + \epsilon) \mqgt$ for some $\epsilon > 0$.
For $\Erw[Z_{k,\ell}(\G,\vec{y})]$ we compute
\begin{align*}
\Erw[Z_{k,\ell}(\G,\vec{y})] \leq
\binom{k}{\ell} \binom{n-k}{k-\ell}
\prod_{i=1}^m \sum_{j=1}^{\vec{y}_{a_i}} \left(
\binom{\Gamma}{j,j,\Gamma-2j}\right. \\
\left.\cdot \left((1-\ell/k)\frac{k}{n}\right)^{2j}
\cdot \left(1-2(1-\ell/k)\frac{k}{n}\right)^{\Gamma-2j}\right).
\end{align*}
The combinatorial meaning is the following:
The binomial coefficients count the number of possible input vectors $\sigma \neq \SIGMA$ of overlap $\ell$ with $\SIGMA$.
The subsequent term measures the probability that a specific such $\sigma$ yields the same results on queries $a_1, \ldots, a_m$ as $\SIGMA$. To see this, we divide the entries $x_1, \ldots, x_n$ into three categories. 
The first category contains those entries that exhibit the same value under $\sigma$ and $\SIGMA$. The second and third category feature those entries that are set to one under $\sigma$ and to zero under $\SIGMA$ and vice versa.
Recall that $\ell$ determines the number of $x_i$ that are set to one under both vectors $\sigma$ and $\SIGMA$.
The probability for a specific entry to be in the 
first category is $1-2(1-\ell/k)k/n$, while the probability for a specific entry to be in the second or third categories is $(1-\ell/k)k/n$ each.
The key observation is that the query results are the same between $\sigma$ and $\SIGMA$ if and only if the number of entries in the second category is identical to the number of entries in the third category.
We compute (a bound on) the sum over the number of entries which are flipped.
Simplifying the term and conditioning on the high probability event $\cR$ 
yields the following lemma.
\begin{lemma}\label{small_overlap_part_function}
For every $0 \leq \ell \leq k- \bc{1 - \exp(-1/2)} \log k$ and a random variable $\vec{X} \sim \Bin_{\geq 1}(\Gamma,2(1-\ell/k)k/n)$, we have
\begin{align*}
\Erw[Z_{k,\ell}] &\leq (1 + O(1))\Erw[Z_{k,\ell}(\G,\vec{y})\mid \cR] \\
&\leq (1+O(1)) \binom{k}{\ell} \binom{n-k}{k-\ell} \bc{ \frac{1}{\sqrt{2 \pi}} \Erw\left[\frac{1}{\sqrt{\vec{X}}} \right]}^m.
\end{align*}
Here, $\Bin_{\geq i}(n, p)$ is the binomial distribution with parameters $n$ and $p$ where we condition that its outcome is at least $i$.
\end{lemma}
\begin{proof}
The product of the two binomial coefficients simply accounts for the number of vectors $\sigma$ that have overlap $\ell$ with $\SIGMA$.
Let $\cS$ denote the event that one specific $\sigma \in \{0,1\}^n$ that has overlap $\ell$ with $\SIGMA$ belongs to $S_{k,\ell}(\G, \vec{y})$.
It suffices to show for $\vec{X} \sim \Bin_{\geq 1}(\Gamma,2(1-\ell/k)k/n)$ that
\begin{align}\label{upper_prob_term}
  \Pr[\cS \mid \cR ] &\leq (1+O(1)) \bc{\frac{1}{\sqrt{2 \pi}} \Erw\left[\frac{1}{\sqrt{\vec{X}}} \right]}^m.
\end{align}
The remainder of the proof is dedicated to showing \cref{upper_prob_term}. 

By the design $\G$, each query contains $\Gamma=n/2$ entries chosen uniformly at random, and we observe that all query results are statistically independent of each other.
Therefore, we need only to determine the probability that for a specific $\sigma$ and a specific query $a_i$ the result is consistent with the result under $\SIGMA$ such that $\vec{y}_i = y_i$.
Given the overlap $\ell$, we know for $\sigma$ drawn uniformly at random that $\pr\brk{\SIGMA_i = \sigma_i = 1 }=\ell/n,$ $\pr\brk{\SIGMA_i = \sigma_i=0}=(n-2k+\ell)/n$ and finally $\pr\brk{\SIGMA_i \neq \sigma_i} =(k-\ell)/n$ holds for all $x_i$, $i=1 \ldots n$.
We get {
\allowdisplaybreaks
\begin{align}
    \Pr[\cS \mid \cR ] 
    &\leq \prod_{i=1}^m \sum_{j=1}^{\vec{y}_{i}} \Bigg(\binom{\Gamma}{j,j,\Gamma-2j}\cdot \left((1-\ell/k)\frac{k}{n}\right)^{2j}\notag
    \\ &\quad\cdot \left(1-2(1-\ell/k)\frac{k}{n}\right)^{\Gamma-2j} \Bigg) \notag \\
     &\leq \Bigg(\sum_{j=1}^{\Gamma/2} \binom{\Gamma}{2j}
     \left(2(1-\ell/k)\frac{k}{n}\right)^{2j}\notag\\
     &\quad \cdot\left(1-2(1-\ell/k)\frac{k}{n}\right)^{\Gamma-2j} \binom{2j}{j} 2^{-2j}\Bigg)^m .\label{EqRandomWalk}
\end{align}
}
The last two components of \eqref{EqRandomWalk} describe the probability that a one-dimensional simple random walk returns to its original position after $2j$ steps, which is by \cref{randomwalk} equal to $(1+O(j^{-1}))/\sqrt{\pi j}$. The former term describes the probability that a $\Bin_{\geq 1}(\Gamma, 2 (1 - \ell/k)k/n))$ random variable $\vX$ takes the value $2j$. For $\ell \leq k- (1-\exp(-1/2))\log k$ the expectation of $\vX$ given $\G$ is at least of order $\log k$ such that the asymptotic description of the random walk return probability is feasible. Note that if $\ell$ gets closer to $k$, the expectation of $\vX$ gets finite, s.t.\ the random walk approximation is not feasible anymore. Therefore, using \cref{binomial}, we can, as long as $\Gamma (2(1-\ell/k)k/n) = \Omega(\log n)$, simplify \eqref{EqRandomWalk} to 
\begin{align*}
    \Pr[ \cS \mid \cR ] &\leq (1 + O(1)) \Bigg(\sum_{j=1}^{\Gamma/2} \binom{\Gamma}{2j} \left(2(1-\ell/k)\frac{k}{n}\right)^{2j} \\
    &\quad\cdot\left(1-2(1-\ell/k)\frac{k}{n}\right)^{\Gamma-2j} \frac{1}{\sqrt{\pi j}}\Bigg)^m \\
    &=(1 + O(1)) \Bigg(\frac{1}{2} \sum_{j=1}^{\Gamma} \binom{\Gamma}{j} \left(2(1-\ell/k)\frac{k}{n}\right)^{j} \\
    &\quad\cdot\left(1-2(1-\ell/k)\frac{k}{n}\right)^{\Gamma-j} \frac{1}{\sqrt{\pi j/2}}\Bigg)^m \\
    &= (1+O(1)) \bc{\frac{1}{\sqrt{2 \pi}} \Erw\left[\frac{1}{\sqrt{\vec{X}}} \right]}^m
\end{align*}
for large $n \gg 1$ which implies \cref{small_overlap_part_function}.
\end{proof}
While the expression given through \cref{small_overlap_part_function} might look hard to work with, it can be simplified using standard asymptotic arguments as follows.
\begin{lemma}\label{small_overlap_part_function_simplify}
For every $0 \leq \ell \leq k- \bc{1 - \exp(-1/2)} \log k$, $m=ck\frac{\log(n/k)}{\log(k)}$ and $n \gg 1$, we have
\begin{align*}
\MoveEqLeft \frac{1}{n}  \log\left( \Erw[Z_{k,\ell}(\G,\vec{y})\mid\cR] \right) \\
&\leq (1+o(1))
\Bigg(\frac kn H\bc{\frac \ell k}+\bc{1-\frac kn} H\bc{\frac{k - \ell}{n-k}} 
\\ & \quad -\frac{c k/n \log (n/k)}{2 \log k} \log\bc{2 \pi \bc{1-\frac \ell k}k}\Bigg).
\end{align*}

\end{lemma}
The key is to choose $c = c(n)$ such that $Z_{k,\ell}(\G,\vec{y})\to0$ for every $\ell\leq k- \bc{1 - \exp(-1/2)} \log k$ when $n\to\infty$. Asymptotically, $\log\left(\Erw[Z_{k,\ell}(\G,\vec{y})]/n\right)$ takes its maximum at $\ell = \Theta \bc{k^2/n}$. Therefore, the r.h.s.\ of \eqref{small_overlap_part_function_simplify} becomes negative if and only if the number of queries $m$ parametrized by $c$ exceeds $\mqgt$. This is formalized in the following lemma and concludes the proof of \cref{Prop_small_overlap}.
\begin{lemma}\label{Lem_small_overlap_partition_function_negative}
For every $0 \leq \ell \leq k- \bc{1 - \exp(-1/2)}\log k$, $0<\theta< 1$ and $\epsilon > 0$ it holds if $m \geq (1+\epsilon)\mqgt$ that
\begin{equation*}
\frac{1}{n} \log \Erw[Z_{k,\ell}(\G,\vec{y})\mid\cR] < 0 .
\end{equation*}
\end{lemma}
\begin{proof}[Proof of \cref{Prop_small_overlap}]
The proposition is a direct consequence of  \cref{small_overlap_part_function,small_overlap_part_function_simplify,Lem_small_overlap_partition_function_negative} and Markov's inequality.
\end{proof}
While we could already establish that there are \whp no feasible vectors $\sigma \in \cbc{0,1}^n$ that have a small overlap with the ground truth $\SIGMA$, we still need to ensure that there are \whp no feasible vectors that have a large overlap with $\SIGMA$.
Indeed, we exclude such vectors with the next proposition.
\begin{proposition}\label{Prop_big_overlap}
Let $\epsilon>0$ and $0<\theta\leq 1$ and assume that $m>(1+\epsilon)\mqgt$. Given $\cR$ we have $Z_{k,\ell}(\G,\vec{y})=0$ for all $k - (1 - \exp(-1/2)) \log k < \ell<k$ \whp
\end{proposition}

The proof is fundamentally easy as it follows the classical coupon collector argument.
However, it needs some technical attention. If we consider a vector $\sigma$ of length $n$ different from $\SIGMA$ with the same Hamming weight $k$, at least one entry that is set to one under $\SIGMA$ is labeled zero under $\sigma$. 
Given the event $\cR$, this entry is part of at least $\vec{\Delta}_i^\star > m/4$ different queries whose results all change by at least $-1$, depending on how often the entry participates.
To compensate for these changes, we need to find $x_1 \dots x_\ell$ that are zero under $\SIGMA$ and one under $\sigma$ such that their joint neighborhood is a super-set of the changed queries.
We show that this only happens with probability $o(1)$ following a classical balls-into-bins argument.
We now give the full technical proof.

\begin{proof}[Proof of \cref{Prop_big_overlap}]
Assume that $\sigma \in \cbc{0,1}^n$ is a second vector that is consistent with the query results  $\vec y$. By definition, there is an index $j \in \cbc{1, \ldots, n}$ for which $\SIGMA(j) = 1$ but $\sigma(j) = 0$. 
By \cref{lem_event_R_high_prob} the size of $\partial^\star x_j$ is at least \[
\vec{\Delta}^{\star}_i \geq \bc{1 - \exp(-1/2)} m - O\bc{\sqrt{m}\log n}
\]
and for any query $a_l \in \partial x_j$ we have $\abs{y_l(\sigma)-y_l (\SIGMA)} \geq 1$.
To guarantee that $y(\SIGMA) = y(\sigma)$ it is necessary to identify a set of $h$ entries $\cX$ for which $\sigma(i) = 1 - \SIGMA(i)$ for all $i \in \cX$ with the property that $\cX \supseteq \partial x_j$.

By construction of $\G$, the number of queries in $\partial^\star x_j$ that do not contain any of the entries in $\cX$, i.e., $\vec{H} = \abs{ \cbc{a \in \partial^\star x_j: \cX \cap \partial a = \emptyset} }$, can be coupled with the number of empty bins in a balls-into-bins experiment as follows.
Given $\G$, throw $b = \sum_{i=1}^h \deg(x_i$) balls into $\deg(x_i)$ bins.
Observe that 
\[
\deg(x_i) \geq \bc{1 - \exp(-1/2)} m - O \bc{\sqrt{m} \log n}
\] and denote by $\vec{H'}$ the number of empty bins in this experiment.
Since for any $x_i$ the $\deg(x_i)$ edges are not only distributed over the $(1 - o(1)) \bc{1 - \exp(-1/2)}m$ query-nodes in $ \partial x_j $ but over all $m$ query-nodes in $\G$, we get 
\begin{align}\label{Eq_CouplingBallsBins}
 \Pr \brk{ \vec{H} = 0 \mid \cR} \leq \Pr \brk{ \vec{H'} = 0 \mid \cR}.
\end{align}
We condition on $\cR$ and therefore $b = (1 + o(1)) h m/2$.
Furthermore, set $L = \log(m) h^{-1}$ and let $\gamma = (1 - \exp(-1/2))$.
Then 
the r.h.s.\ of \eqref{Eq_CouplingBallsBins} becomes
\begin{align*}
   \Pr \brk{ \vec{H'} = 0 \mid \cR} &
   \leq \bc{1 - \bc{ 1 - \frac{1}{\gamma m}}^{hm/2} }^{\gamma m} \\
   &= (1 + o(1)) \exp \bc{ - \gamma m^{ 1 - L/(2 \gamma) } }.
\end{align*}
Therefore, if $L < 2 \gamma$, or equivalently, 
\begin{align*}
 h < 2 \gamma \log \bc{ m } \sim 2 \gamma \bc{\log k + \log \log k}, \\
\lefttag{we have}   \Pr \brk{ \vec{H'} = 0 \mid \cR} \leq n^{- \omega(1)}.
\end{align*}
Thus, a Hamming distance of at least one between $\SIGMA$ and $\sigma$ immediately implies that the Hamming distance is at least $2 \gamma \bc{\log k + \log \log k}$ with probability $1 - n^{- \omega(1)}$. 
A union bound over all $k$ one-entries implies the proposition.
\end{proof}

\begin{proof}[Proof of \Cref{Thm_inf}]
The theorem follows directly from \cref{Prop_small_overlap,Prop_big_overlap}.
\end{proof}

\begin{figure*}[t]
\setcounter{figure}{2}
\input{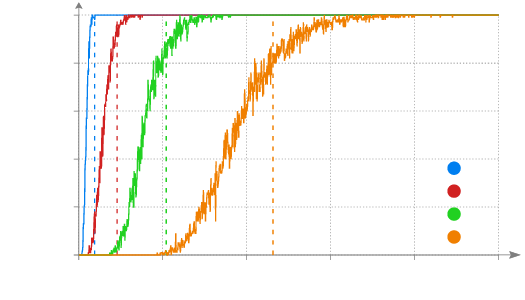}
\hfill
\input{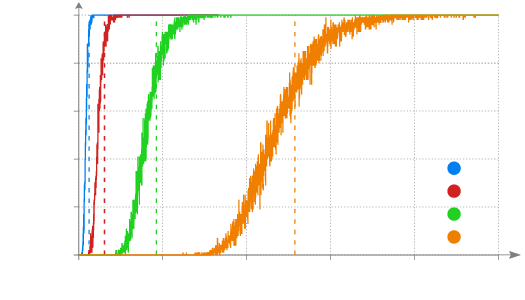}
\vspace{0.5ex}
\caption{The plot shows the rate of successful recovery of $\SIGMA$ among $100$ independent simulation runs over the number of queries $m$ for different values of $\theta$ and $n =  10^3$ (left) and $n = 10^4$ (right).
}
\label{fig:success-rate}
\end{figure*}

\begin{figure*}[t]
\setcounter{figure}{3}
\input{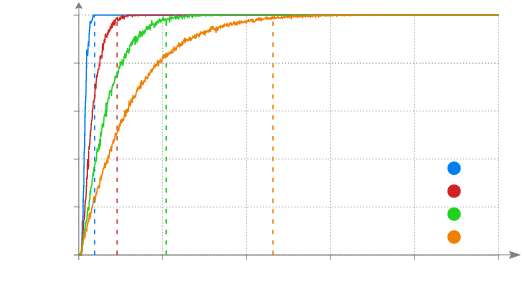}
\hfill
\input{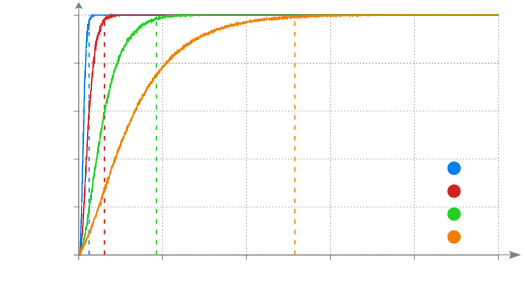}
\vspace{0.5ex}
\caption{The plots show the \emph{overlap} -- the fraction of correctly classified one-entries -- among $100$ independent simulation runs over the numbers of queries $m$ for different values of $\theta$ and $n =  10^3$ (left) and $n = 10^4$ (right).
}
\label{fig:overlap}
\end{figure*}

\begin{figure}[H]
\setcounter{figure}{1}
\input{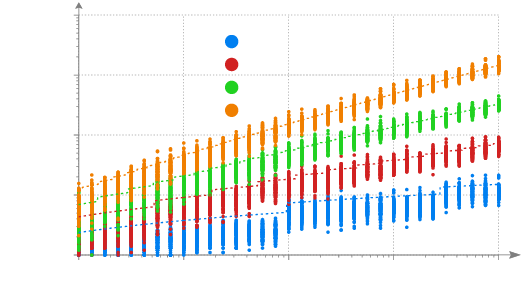}
\vspace{0.5ex}
\caption{The required number of queries until $\SIGMA$ can be exactly reconstructed for different vector lengths $n$ and $\theta$ regimes. For each value of $n$, 100 simulations were carried out independently.\label{fig:required-number-of-tests-1}}
\end{figure}

\section{Empirical Analysis and Simulation Results}

In this section we present simulation results for the MN-Algorithm (\cref{MN_algorithm}).
{
\newcommand\CC{C\nolinebreak[4]\hspace{-.05em}\raisebox{.4ex}{\relsize{-2}{\textbf{++}}}}
Our simulation software is implemented in the \CC{} programming language.
It performs a faithful simulation of the parallel system.
To generate the random structures, we resort to the Mersenne Twister \texttt{mt19937\_64} as provided by the \CC11 \texttt{\textless random\textgreater{}} library.
All of our simulations have been carried out on machines equipped with 20 Intel(R) Xeon(R) E5-2630 v4 CPU cores, backed by 128GiB memory, and running the linux 5.11 kernel.
All required code to reproduce our figures, including the gnuplot scripts and various helper tools, can be obtained from our public github repository. 
}

In our first empirical result in \cref{fig:required-number-of-tests-1} we analyze the number of queries required to reconstruct $\SIGMA$ for $n \in [10^2, 10^6]$ and different values of $\theta$.
The dotted lines show our theoretical asymptotic bounds.
Note that the discontinuities in the theoretical bound stem from rounding the number of one-entries $k$ to the closest integer. We remark that our simulation results align well with the theoretical predictions for larger values of $n$.
For smaller values of $n$, our theoretical results are too optimistic:
the lower-order term hidden in the $o(1)$ in \cref{eqs_greedy} scales as $\Theta \bc{ \frac{\sqrt{\log n}}{k} }$, and while this expression decreases polynomially fast in $n$, it is far from vanishing for small values of $n$ and $\theta$.

In \cref{fig:success-rate,fig:overlap} we analyze the success probability for exact reconstruction of $\SIGMA$ and the number of correctly identified one-entries.
For different numbers of queries  we conducted 100 independent simulation runs for $n = 10^3$ and $n = 10^4$ and different values of $\theta$.
The dashed lines show the phase-transitions predicted by \cref{Thm_alg}.
The data in \cref{fig:overlap} indicate that all but a small fraction of one-entries are correctly detected, even if the exact reconstruction of $\SIGMA$ is still quite unlikely according to \cref{fig:success-rate}.
Overall, the implementation hints at the practical usability of the MN-Algorithm, even for small values of $n$.

\begin{remark}
The formal proof of the algorithmic bound directly gives an insight about the convergence speed and thus about the expected performance of the MN-Algorithm for finite~$n$:
 we can compute that the MN-Algorithm requires an additional multiplicative factor of at least \begin{equation*} \bc{ 1 + {\sqrt{2} \log n} \bc{ { 4 \bc{1 - \exp \bc{-1/2}} m k } }^{-1/2}} \end{equation*} queries in addition to the asymptotic analysis for $n \rightarrow \infty$.
This explains the (slight) deviation of the theoretical and the empirical results for small values of $n$. See the proof of \cref{Cor_tails} in \cref{sec:proof_algo} for the rigorous analysis.
\end{remark}

\section{Conclusions and Open Problems}

In this paper we analyze the binary pooled data problem with additive queries both from an information-theoretic and an algorithmic point of view.
Our first result is a simple greedy reconstruction scheme that performs well even close to the information-theoretic boundaries.
Our main concern is the design of a reconstruction scheme that works well when all queries are conducted in parallel.
In a series of simulations we show that this scheme is applicable to a large range of parameters that can be expected from real-world instances.
For example, our data indicate that on average we correctly identify  99\% of the one-entries when conducting only 220 queries for $n = 1000$ and $\theta = 0.3$.
Our second result sheds light on the information-theoretic achievability threshold, where our theorem closes the open gap between the results of  \cite{djackov_1975} and \cite{grebinski_2000} by establishing a sharp phase transition.

An immediate open problem is to close the gap between the algorithmic and the information theoretic threshold. {Furthermore, there are similar reconstruction problems in which parallel conductance of all queries is crucial. As discussed in the introduction, group testing is such a prime example which was recently fully understood using similar techniques as in the present work. A less well understood reconstruction problem is \emph{threshold group testing} \cite{Chan2013, DeMarco2020}, in which a query outputs $1$ if and only if the number of positive entries exceeds a threshold $T > 0$. It is very likely that the techniques of the present contribution can be applied to threshold group testing as well, as they were previously applied to various reconstruction problems, but the tailor-made application remains a highly non-trivial challenge.}
Another exciting avenue for future research are partially parallelizable designs.
Suppose that, for instance, $L$ processing units can be used to evaluate queries in parallel.
Then it is a natural requirement for a design to always conduct up to $L$ queries in parallel.
An interesting open question then is to analyze the trade-offs that arise in such partially parallelized schemes.
In particular, there might be designs providing efficient reconstruction algorithms that outperform the completely parallel design studied in this paper.

\section*{Acknowledgements}
The authors thank Uriel Feige for various detailed comments which improved the quality of the paper significantly. Furthermore, the authors thank Petra Berenbrink and Amin Coja-Oghlan for helpful discussions and important hints.

\bstctlcite{IEEEexample:BSTcontrol}
\bibliography{bibliography}
\appendix






\medmuskip=3.0mu plus 2.0mu minus 4.0mu
\thinmuskip=2.0mu
\thickmuskip=2.0mu plus 5.0mu

\let\section\subsection
\let\subsection\paragraph

We start this appendix with standard concentration bounds that we use throughout this paper.
\begin{lemma}[\cite{Asymptopia}]\label{Lem_ChernoffBoundDKL}
Let $X \sim \Bin(n, p)$ and $\delta \in (0,1)$.
  \begin{align*}
    \lefttag{Then} \Pr \brk{ X > (1+\delta)np} &\leq \exp \bc{ - np \delta^2 / (2 + \delta) } \\
    \lefttag{and}  \Pr \brk{ X < (1-\delta)np} &\leq \exp \bc{ - np \delta^2 / 2 }.
  \end{align*}
\end{lemma}
For binomial random variables, the Jensen gap provides good approximations.
\begin{lemma}[follows from Eq. (1.1) of \cite{gao_2018}]\label{Cor_JensenGap}
Let $\Bin_{\geq i}(n, p)$ be the binomial distribution with parameters $n$ and $p$ where we condition that its outcome is at least $i$.
Let $\vX \sim \Bin_{x \geq 1}(n,p)$ with $np \to \infty$.
Then, for $\ell \in \cbc{ 1/2, 1 }$, we have
\begin{align*}
    \Erw\left[ \vX^{- \ell} \right] &= \bc{ 1 + o(n^{-1}) } {\Erw[\vX]}^{- \ell}. 
\end{align*}
\end{lemma}
The following lemmas are results on random walks. 
\begin{lemma}[\cite{Asymptopia}, Section 1.5]\label{randomwalk}
The probability that a simple random walk on $\ZZ$ with $2j$ steps will end at its original position is given by $(\pi j)^{-1/2} + O(j^{-3/2}).$
\end{lemma}
{ \allowdisplaybreaks
\begin{lemma}\label{binomial}
The following asymptotic equivalence holds for every $0 < p = p(n) < 1$ when $np \to \infty$.
\begin{align*}
    \MoveEqLeft \sum_{j=1}^{n/2} \binom{n}{2j}p^{2j}(1-p)^{n-2j}j^{-1/2}  \\
    &= 2^{-1/2} \sum_{j=1}^n \binom{n}{j} p^{j}(1-p)^{n-j} j^{-1/2} + O( (np)^{-1} )
\end{align*}
\end{lemma}
}
\begin{proof}
Let $\vX \sim \Bin_{\geq 1}(n, p)$ and define $a_j = \Pr \bc{\vX = j} / \sqrt{j/2}$ for $j=1 \dots n$. Then $$a_{j+1}/a_j = \bc{p/(1-p)}\bc{j / (j+1)^3}^{1/2} \bc{n - j}$$ is larger than $1$ up to $j^{\star} \in \cbc{ \floor{(n+1)p}, \floor{(n+1)p -1}}$, depending on $n$ being even or odd, and strictly less than $1$ for $j = j^{\star}+1,...,n$. Furthermore, $a_j = o(1)$ for every $j$. Define $j'$ as the largest even integer s.t.\ $ j' \leq j^{\star}$. Then
\begin{align}
    \sum_{j = 1}^{n/2} a_{2j} & \geq \frac{1}{2} \bc{\sum_{j = 1}^{j'/2} a_{2j} + a_{2j -1} + \sum_{j = j'/2 + 1}^{n/2 - 1} a_{2j} + a_{2j + 1}} \notag \\
    &= \bc{ \frac{1}{2} \sum_{j=1}^n a_j} + O((np)^{-2}), \notag 
\end{align}
The upper bound follows similarly, and together they imply the lemma. 
%
\end{proof}



We now prove the concentration results for the random regular pooling design.

\begin{proof}[Proof of \cref{lem_event_R_high_prob}]
\label{proof_properties_pooling_scheme}
Fix an index $i \in [n]$.
From the construction of $\G$ it follows that $\vec{\Delta}_i$ is distributed as $\Bin(mn/2, 1/n)$.
Then \cref{Lem_ChernoffBoundDKL} implies
\begin{align*}
   \Pr \bc{\vec{\Delta}_i > m/2 + O\bc{\sqrt{m} \log^2 n}} = n^{-\omega(1)}.
\end{align*}
Furthermore, the probability that an entry $x_i$ is contained in a specific query $a_j$ is given by
\begin{align*}
    p = 1 - \bc{1 - n^{-1}}^{\Gamma} = \bc{1+n^{-\Omega(1)}} \bc{1 - 1/\sqrt{e}}.
\end{align*}
Since queries select their participating entries independently of each other, we observe that $\vec{\Delta}^\star_i \sim \Bin \bc{m, p}$. Thus, \cref{Lem_ChernoffBoundDKL} implies
\begin{align*}
   \Pr \bc{\vec{\Delta}^\star_i > \bc{1-1/\sqrt{e}} m + O\bc{\sqrt{m} \log n}} = n^{-\omega(1)}.
\end{align*}
The union bound over all $n$ entries concludes the proof.
\end{proof}

\begin{proof}[Proof of \cref{small_overlap_part_function_simplify}]
Let $\vX \sim \Bin_{\geq 1}(\Gamma, 2 (1 - \ell/k)k/n)$. Then 
\begin{align}\label{expectation_upper}
\Erw[Z_{k,\ell}(\G,\vec{y})\mid \cR] \leq O(1) \cdot \binom{k}{\ell} \binom{n-k}{k-\ell} \bc{\frac{1}{2 \pi \Erw[\vec{X}]}}^{\frac m 2}
\end{align}
by \cref{small_overlap_part_function,Cor_JensenGap}.
We use the well known fact \cite{Asymptopia}
that as $n\to\infty$ we have for $p \in (0,1)$ that
\begin{align*}
n^{-1} \log \binom{n}{np} \to H(p) := -p\log (p)-(1-p) \log (1-p).
\end{align*}
We apply the $\log(\cdot)$ to \eqref{expectation_upper} and divide it by $n$. Then we calculate using $m = c k \log(n/k) \log^{-1}(k)$
\begin{align*}
    \MoveEqLeft n^{-1} \log\left(\Erw\left[Z_{k,\ell}(\G,\vec{y})\mid \cR  \right]\right)\notag \\
    &\leq (1+o(1))\Bigg(\frac kn H\bc{\frac \ell k}+\bc{1-\frac kn} H\bc{\frac{k-\ell}{n-k}} \notag\\
    &\quad-\frac{c k/n \log (n/k)}{2 \log k} \log\bc{2 \pi k \bc{1- \ell/k}}\Bigg).
    \qedhere
\end{align*}
\end{proof}

\mbox{}

\begin{proof}[Proof of \cref{Lem_small_overlap_partition_function_negative}]
Let $\gamma = 1 - \exp(-1/2)$ and recall  
\begin{align}
    \label{eq_recall} m = c k \log \bc{\frac {n}{k}} / \log k = c \frac{1 - \theta}{\theta} k ~ \text{and} ~ 0 \leq \ell \leq k - \gamma \log k 
\end{align}
for a constant $c>0$. Then define $f_{n,k}: [0, k - \gamma \log k] \to \RR$ as
\begin{align}
 \ell \mapsto \Bigg(\frac kn H\bc{\frac \ell k}+\bc{1-\frac kn} H\bc{\frac{k-\ell}{n-k}}\notag\\-\frac{c k/n \log (n/k)}{2 \log k} \log(2 \pi (1-\ell/k)k )\Bigg) \label{func_f}
\end{align}
and assume, as usual, $0 \log 0 = 0$.
By \cref{small_overlap_part_function_simplify} we get \[
n^{-1} \log\left(\Erw\left[Z_{k,\ell}(\G,\vec{y})\mid \G \right]\right) \leq (1 + o(1)) f_{n,k}(\ell). \] Expanding the entropy yields
{\small
\begin{align*}
     f_{n,k}(\ell) &=\frac 1n \Bigg(-\ell \log\bc{\frac \ell k} -\bc{k-\ell}\log\bc{1-\frac \ell k}\\
    &\quad-\bc{k-\ell} \log\bc{\frac{k-\ell}{n-k}}-\bc{n-2k+\ell} \log \bc{1-\frac{k-\ell}{n-k}}\\
    &\quad +\frac{c k \log (k/n)}{2 \log k} \log \bc{2 \pi k\bc{1-\frac \ell k}}\Bigg), \\
    f'_{n,k}(\ell) &= \frac 1 n \Bigg(-\log \bc{\frac \ell k} + \log\bc{1-\frac \ell k} + \log \bc{\frac{k-\ell}{n-k}}\\
    &\quad-\log\bc{1-\frac{k-\ell}{n-k}}-\frac{c k \log(k/n)}{2 (k-\ell) \log k }  \Bigg), \quad \text{and} \\
    f''_{n,k}(\ell) &= \frac 1 n \Bigg(-\frac 1\ell -\frac{2}{k-\ell}-\frac{1}{n-2k+\ell} - \frac{c k \log(k/n)}{2 \log k (k-\ell)^2} \Bigg).
\end{align*}
}
If $\ell = o(k)$ we get $\abs{\frac{1}{k - \ell} \bc{ 2 - \frac{c k (1 - \theta)}{2 \theta (k - \ell)} }} \ll \frac{1}{\ell}$ and therefore
\begin{align*}
    n f_{n,k}''(\ell) = - \frac{1}{\ell} - \frac{1}{n - 2k + \ell} - \frac{1}{k - \ell} \bc{ 2 - \frac{c k (1 - \theta)}{2 \theta (k - \ell)} } < 0.
    \end{align*}
This shows that $f_{n,k}'$ is monotonically decreasing in $\ell$ for large enough $n$. Furthermore, $f_{n,k}'$ is continuous on $(0,  k - \gamma \log k]$. Let $\tilde c > 0$ be an arbitrary constant. Then
\begin{align*}
    n f_{n,k}' \bc{\tilde c \frac{k^2}{n}} 
     & =- \log \bc{ \tilde c } + \frac{c (1 - \theta)}{\theta} + o(1) .
\end{align*}
This implies that there are $0 < \tilde c_1 < \tilde c_2 < \infty$ s.t.
\[ n f_{n,k}' \bc{\tilde c_1 \frac{k^2}{n}} > 0 \quad \text{and} \quad n f' \bc{\tilde c_2 \frac{k^2}{n}} < 0. \]
By the intermediate value theorem it follows that there is $\hat c \in [\tilde c_1, \tilde c_2]$ s.t. $\hat c \frac{k^2}{n}$ is the unique maximizer of $f_{n,k}$ for $\ell = o(k)$. Finally, by putting this value into \cref{func_f} we obtain that the highest order terms satisfy
\begin{align}
   n f_{n,k}\bc{ \hat c \frac{k^2}{n} } < 0 
    &\Longleftrightarrow     c > -2\frac{H(k/n)}{k/n \log (k/n)} = 2 + o(1). \label{eq_cond_c}
\end{align}
Furthermore, if $k - \gamma \log k \geq \ell = \Theta(k)$, we get \begin{align}
    \label{eq_linear} n f_{n,k}(\ell) = - \frac{c (1 - \theta)}{2 \theta} k \log (k) + O(k)
\end{align}
by definition, which is negative. Therefore, the lemma follows from \cref{eq_recall,eq_cond_c,eq_linear}.
\end{proof}

\end{document}